\documentclass[twocolumn,pra,twocolumn,superscriptaddress]{revtex4}
\usepackage{color}
\usepackage{amsmath}
\usepackage{amssymb}
\usepackage{algorithm}
\usepackage{algorithmic}
\usepackage{soul}
\usepackage{braket}
\usepackage{graphicx}
\usepackage{tikz}
\usepackage{mathrsfs} 
\usepackage{float}
\usepackage{bbm}
\usepackage{epstopdf}
\usepackage{epsfig}
\usepackage{verbatim}
\usepackage{array}
\usepackage{setspace}
\usepackage{times}
\usepackage{amsmath, amssymb, amsthm}

\newtheorem{theorem}{Theorem}

\usepackage[unicode=true,
 bookmarks=true,bookmarksopen=false,
 breaklinks=false,pdfborder={0 0 1},backref=false,colorlinks=true]
 {hyperref}
\hypersetup{
 linkcolor=magenta, urlcolor=blue, citecolor=blue, pdfstartview={FitH}, hyperfootnotes=false, unicode=true}

\newcolumntype{C}[1]{>{\centering\arraybackslash$}p{#1}<{$}}

\begin{document}

\title{Efficient Large-Scale Quantum Optimization via Counterdiabatic Ansatz}

\author{Jie Liu}
\affiliation{Department of Physics, City University of Hong Kong, Tat Chee Avenue, Kowloon, Hong Kong SAR, China, and City University of Hong Kong Shenzhen Research Institute, Shenzhen, Guangdong 518057, China}
\author{Xin Wang}
\email{x.wang@cityu.edu.hk}
\affiliation{Department of Physics, City University of Hong Kong, Tat Chee Avenue, Kowloon, Hong Kong SAR, China, and City University of Hong Kong Shenzhen Research Institute, Shenzhen, Guangdong 518057, China}
\affiliation{Quantum Science Center of Guangdong-Hong Kong-Macao Greater Bay Area, Shenzhen, Guangdong 518045, China}

\date{\today}

\date{\today}

\begin{abstract}
Quantum Approximate Optimization Algorithm (QAOA) is one of the fundamental variational quantum algorithms, while a version of QAOA that includes counterdiabatic driving, termed Digitized Counterdiabatic QAOA (DC-QAOA), is generally considered to outperform QAOA for all system sizes when the circuit depth for the two algorithms are held equal. Nevertheless, DC-QAOA introduces more CNOT gates per layer, so the overall circuit complexity is a tradeoff between the number of CNOT gates per layer and the circuit depth and must be carefully assessed. In this paper, we conduct a comprehensive comparison of DC-QAOA and QAOA on MaxCut problem with the total number of CNOT gates held equal, and we focus on one implementation of counterdiabatic terms using nested commutators in DC-QAOA, termed as DC-QAOA(NC). We have found that DC-QAOA(NC) reduces the overall circuit complexity as compared to QAOA only for sufficiently large problems, and for MaxCut problem the number of qubits must exceed 16 for DC-QAOA(NC) to outperform QAOA. Additionally, we benchmark DC-QAOA(NC) against QAOA on the Sherrington-Kirkpatrick model under realistic noise conditions, finding that DC-QAOA(NC) exhibits significantly improved robustness compared to QAOA, maintaining higher fidelity as the problem size scales. Notably, in a direct comparison between one-layer DC-QAOA(NC) and three-layer QAOA where both use the same number of CNOT gates, we identify an exponential performance advantage for DC-QAOA(NC), further signifying its suitability for large-scale quantum optimization tasks. We have further shown that this advantage can be understood from the effective dimensions introduced by the counterdiabatic driving terms. Moreover, based on our finding that the 
 optimal parameters generated by DC-QAOA(NC) strongly concentrate in the parameter space, we have devised an instance-sequential training method for DC-QAOA(NC) circuits, which, compared to traditional methods, offers performance improvement while using even fewer quantum resources. Our findings provide a more comprehensive understanding of the advantages of DC-QAOA circuits and an efficient training method based on their generalizability.

\end{abstract}

\maketitle

\section{INTRODUCTION} \label{Intro}

Variational Quantum Algorithms (VQAs) provide an efficient way to tackle complex problems by merging classical optimization techniques with quantum computing circuits  \cite{cerezoVariationalQuantumAlgorithms2021b,kandalaHardwareefficientVariationalQuantum2017a,mitaraiQuantumCircuitLearning2018,mccleanTheoryVariationalHybrid2016,liuVariationalQuantumEigensolver2019,peruzzoVariationalEigenvalueSolver2014a,haugCapacityQuantumGeometry2021a,abbasPowerQuantumNeural2021}. Inspired by adiabatic quantum processes  \cite{comparat2009general}, the Quantum Approximate Optimization Algorithm (QAOA) is a type of VQA particularly suited for combinatorial optimization problems  \cite{hadfieldQuantumApproximateOptimization2019a,zhouQuantumApproximateOptimization2020,zhouQAOAinQAOASolvingLargeScale2023}. QAOA constructs a parametrized quantum circuit (``ansatz'') composed of alternating layers of problem-specific cost functions (encoded as quantum operators) and mixing unitary operations. These parameters are then optimized iteratively using classical optimization methods to approach optimal solutions to a given problem. The capability of QAOA to handle a wide array of optimization problems, including MaxCut, Graph Partitioning, and others, underscores its significance in quantum computing  \cite{zhuAdaptiveQuantumApproximate2022,ruanQuantumApproximateOptimization2023,ohSolvingMultiColoringCombinatorial2019}.

In principle, QAOA becomes more powerful as the number of layers increases \cite{farhiQuantumApproximateOptimization2014}. However, its implementation is challenging for current Noisy Intermediate-Scale Quantum (NISQ) devices due to their limited coherence times. On the other hand, adiabatic quantum processes can be accelerated by adding counterdiabatic terms (cd-terms), known as counterdiabatic driving (cd-driving) in Shortcut to Adiabaticity (STA)  \cite{chenFastOptimalFrictionless2010,vacantiTransitionlessQuantumDriving2014,jiCounterdiabaticTransferQuantum2022,hartmannRapidCounterdiabaticSweeps2019a}. It has been demonstrated that, under certain situations, the performance of QAOA can be enhanced by incorporating cd-terms in a digitized fashion, resulting in the Digitized Counterdiabatic Quantum Approximate Optimization Algorithm (DC-QAOA)  \cite{hegadeShortcutsAdiabaticityDigitized2021,chandaranaDigitizedcounterdiabaticQuantumApproximate2022}.

Recent studies have compared the performance of DC-QAOA and QAOA in a variety of problems, including ground state preparation  \cite{chandaranaDigitizedcounterdiabaticQuantumApproximate2022}, molecular docking  \cite{dingMolecularDockingQuantum2023a}, protein folding  \cite{chandaranaDigitizedCounterdiabaticQuantumAlgorithm2023}, and portfolio optimization  \cite{hegadePortfolioOptimizationDigitized2022}. In all these studies, the comparisons were made when both algorithms have the same number of circuit layers. The fact that DC-QAOA exhibits an advantage over QAOA in these studies suggests that cd-terms effectively reduce the number of circuit layers required to achieve the same level of efficiency. Nevertheless, the circuit complexity for a NISQ device, which is directly linked to the execution time,  depends not only on the number of layers but also the number of two-qubit gates in each layer. While reducing the number of layers required by DC-QAOA to achieve the same efficiency as QAOA, cd-terms introduce more two-qubit gates, substantially complicating the circuit. Therefore, apart from the number of layers considered in the literature, the number of two-qubit gates is also playing a key role in comparing the performances between DC-QAOA and QAOA on actual NISQ devices.

In this paper, we compare DC-QAOA and QAOA on the MaxCut problem  \cite{hadlock1975findinga}, where the number of layers is increased until both algorithms reach the predefined fidelity to the ground state of optimization problems. We use the number of CNOT gates as an indicator of efficiency since it characterizes the difficulty of implementing such algorithms on quantum computers. Using MaxCut on randomly generated unweighted graphs as benchmarks, we found that for graphs with over 16 vertices, DC-QAOA can achieve better performance using fewer two-qubit gates, while QAOA is better for smaller graphs. By closely examining the effective dimension of the two circuits, we find that adding cd-terms in each layer provides more effective dimensions than simply adding more layers to QAOA, and the cd-terms can reduce non-diagonal transitions, thereby improving DC-QAOA's performance. Furthermore, the optimal parameters obtained by training DC-QAOA display a strong concentration in the parameter space, demonstrating potential transferability and generalizability from optimization problems on a smaller scale to those on a larger scale. Based on our findings, we devised an Instance-Sequential Training (IST)  method for DC-QAOA, which enables comparable or superior outcomes with substantially reduced quantum resources. As a result, DC-QAOA emerges as a more viable option for extensive quantum optimization tasks, in contrast to QAOA, which may be more efficient for smaller-scale problems.

The remainder of the paper is organized as follows. In Section \ref{definition}, we introduce the formalism including definitions of QAOA and DC-QAOA. Section \ref{results} provides the results. We show a comprehensive comparison of DC-QAOA and QAOA in Section \ref{scaling}. In Section \ref{suppression}, we discuss how DC-QAOA benefits from cd-terms, and in Section \ref{training}, we demonstrate the parameter concentration phenomenon in DC-QAOA and propose a quantum resource-efficient training method. We conclude in Section \ref{conclude}.

\begin{figure}
    \centering
    \includegraphics[width=\linewidth]{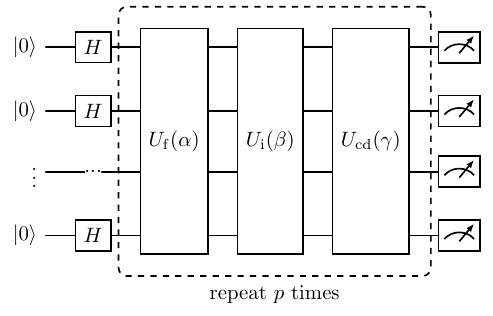}
    \caption{Schematic circuit diagram for DC-QAOA where $\alpha$, $\beta$, $\gamma$ are optimized by a classical optimizer based on a loss function. Here $U_{\mathrm{cd}}(\gamma)$ acts as the cd-term to suppress diabatic transitions. The parameterized block is repeated $p$ times, where $p$ is number of layers of DC-QAOA.}
    \label{fig:Fig1}
\end{figure}

\section{Formalism} \label{definition}

QAOA can be seen as a digitized version of Adiabatic Quantum Computing with trainable parameters  \cite{barendsDigitizedAdiabaticQuantum2016}. The ansatz of QAOA stems from the quantum annealing Hamiltonian:
\begin{equation}
H(t) = [1-\lambda(t)]H_\mathrm{i} + \lambda(t)H_\mathrm{f},
\end{equation}
where $\lambda(t) \in [0,1]$ is the annealing schedule in $t \in [0,1]$. As $\Dot{\lambda} \rightarrow 0$, the ground state of $H_\mathrm{i}$ can evolve adiabatically into the ground state of $H_\mathrm{f}$. Here, $H_\mathrm{i} = \sum_k h_k\sigma_k^x$, allowing the ground state of $H_\mathrm{i}$ to be easily prepared as an equally weighted superposition state in the computational bases, and $H_\mathrm{f}$ can be any Hamiltonian of interest, usually called problem Hamiltonian. Quantum annealing can be digitized with trotterized time evolution:
\begin{equation}
    U(T) = \prod_{j=1}^{p}e^{-i[1-\lambda(j\Delta t)]H_\mathrm{i}\Delta t}e^{-i\lambda(j\Delta t)H_\mathrm{f}\Delta t}.
\end{equation}
Here, $p$ is the number of trotter layers, and when $p$ is sufficiently large, $U(T)$  approaches the adiabatic limit. $U(T)$ can be further parameterized with a set of trainable parameters $(\alpha_1,\alpha_2,...,\alpha_p;\beta_1,\beta_2,...,\beta_p)$ instead of original annealing schedule to counter coherent errors in NISQ devices:
\begin{equation}
    \begin{split}
        U(\Vec{\alpha},\Vec{\beta})=&U_\mathrm{i}(\beta_p)U_\mathrm{f}(\alpha_p)U_\mathrm{i}(\beta_{p-1})U_\mathrm{f}(\alpha_{p-1})...\\
        &\times U_\mathrm{i}(\beta_1)U_\mathrm{f}(\alpha_1),
    \end{split} \label{eq:QAOA}
\end{equation}
where the unitary operators are $U_\mathrm{i}(\beta)=\exp(-i\beta H_\mathrm{i})$ and $U_\mathrm{f}(\alpha)=\exp(-i\alpha H_\mathrm{f})$. By maximizing an objective function $C(\Vec{\alpha},\Vec{\beta})$, defined below in Eq.~\eqref{eq:costfunc} for the MaxCut problem, we minimize the distance between the output state $\ket{\psi_\mathrm{o}}$ and the exact ground state $\ket{\psi_\mathrm{g}}$ of $H_\mathrm{f}$. Hence adiabatic quantum evolution is transformed into an optimization problem with $2p$ parameters.

However, the requirement for a large number of trotter layers $p$ in QAOA (sometimes referred to as ``standard QAOA'' hereafter) to approximate adiabatic evolution poses challenges due to the increased circuit complexity. To address this limitation, DC-QAOA has been developed  \cite{chandaranaDigitizedcounterdiabaticQuantumApproximate2022}. DC-QAOA builds upon analog cd-driving which can shorten the evolution time of adiabatic transitions by adding an extra cd-term, presenting an opportunity to effectively reduce number of layers $p$ in QAOA.
The nature of the cd-terms depends on the choice of appropriate adiabatic gauge potential \cite{Sels}. In the realm of analog cd-driving, the commonly utilized gauge potential is first-order nested commutators
\begin{equation}
    H_{\mathrm{cd}} = i\Gamma(t)[H,\partial_\lambda H].
\end{equation}
In cd-driving, the effective Hamiltonian can be written as 
\begin{equation}
    H_{\mathrm{eff}}(t)=[1-\lambda(t)]H_\mathrm{i}+\lambda(t)H_\mathrm{f}+\Dot{\lambda}H_{\mathrm{cd}}.
\end{equation}
Following the same routine in QAOA, cd-driving can be similarly digitized and parameterized as 
\begin{equation} 
    \begin{split}
        U(\Vec{\alpha},\Vec{\beta},\Vec{\gamma})=&U_{\mathrm{cd}}(\gamma_p)U_\mathrm{i}(\beta_p)U_\mathrm{f}(\alpha_p)\\
        &\times U_{\mathrm{cd}}(\gamma_{p-1})U_\mathrm{i}(\beta_{p-1})U_\mathrm{f}(\alpha_{p-1})...\\
        &\times U_{\mathrm{cd}}(\gamma_1)U_\mathrm{i}(\beta_1)U_\mathrm{f}(\alpha_1).
    \end{split} \label{eq:DC-QAOA}
\end{equation}
A schematic circuit diagram of a typical $p$-layer DC-QAOA is shown in Fig.~\ref{fig:Fig1}. The detailed implementation of three possible cd-terms in DC-QAOA using elementary gates, including first-order nested commutator (hereafter referred to as NC) \cite{chandaranaDigitizedCounterdiabaticQuantumAlgorithm2023}, single qubit Y-rotation (referred to as Y)  \cite{ieva2023COLD}, and two-qubit YY-Interaction (referred to as YY)  \cite{chaiShortcutsQuantumApproximate2022}, can be seen in Fig.~\ref{fig:Fig2}.

To explore the advantage of DC-QAOA over QAOA on quantum optimization problems, we use the MaxCut problem as a benchmark \cite{goemans1995improved}. For the MaxCut problems with an unweighted graph $\mathcal{G}=(V,E)$, where $V$ and $E$ are the vertex and edge sets respectively, the objective function $C(\emph{z})$ is defined on a string of binary values $\emph{z} = (z_1,z_2,...,z_n)$:
\begin{equation}
    C(z) = \frac{1}{2}\sum_{(j,k)\in E}(1-z_j z_k), \label{eq:costfunc}
\end{equation}
which aims to separate the vertices into two subsets so that $C(\emph{z})$ is maximized. Depending on which subset vertex $j$ is in, $z_j$ is assigned a binary value of 0 or 1. The solution is embedded in the ground state of $H_\mathrm{f}$, written as
\begin{equation}
    H_\mathrm{f} = \sum_{(j,k)\in E}\sigma_j^z\sigma_k^z.
\end{equation}
By finding the ground state $\ket{\psi_\mathrm{g}}$ of $H_\mathrm{f}$, we can maximize the objective function $C(\emph{z})$.
\begin{figure}
    \centering
    \includegraphics[width=\linewidth]{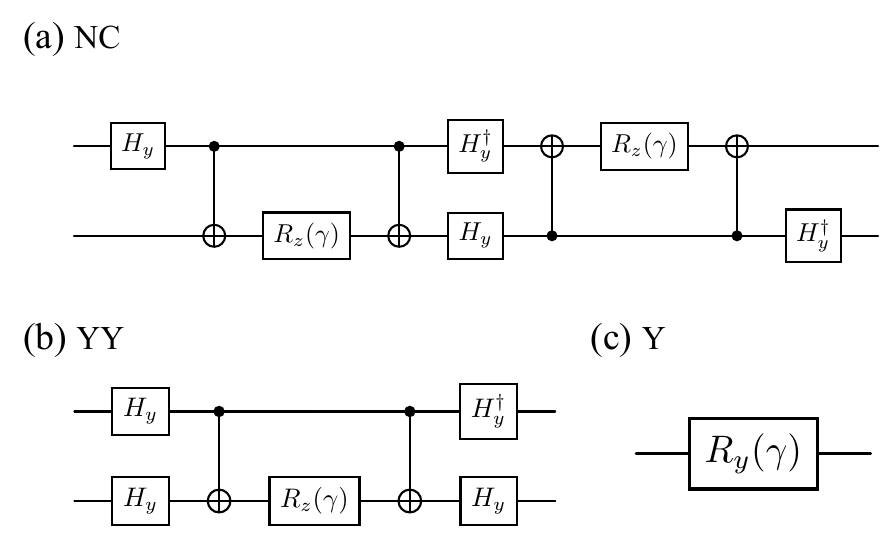}
    \caption{Three possible implementations of the cd-term $U_\mathrm{cd}(\gamma)$ using CNOT gates and single qubit rotation in DC-QAOA considered in this paper. (a) NC: First order nested commutator based on $H_\mathrm{i}$ and $H_\mathrm{f}$ in MaxCut. (b) YY: YY-Interaction. (c) Y: Single qubit Y-rotation. $R_z(\gamma)$ and $R_y(\gamma)$ are single qubit rotations around z- and y-axis respectively, where $\gamma$ is the rotational angle. $H_y=\frac{1}{\sqrt{2}}(Y+Z)$.  }
    \label{fig:Fig2}
\end{figure}

In this study, for each given number of vertices between 9 and 16, we randomly generate one graph, subject to the condition that the edges between every pair of vertices are connected with probability $0.5$. These graphs are shown in the Appendix.

\section{RESULTS} \label{results}

\subsection{Circuit Length Scaling Behavior} \label{scaling}

In the current literature, the comparison between DC-QAOA and QAOA is done holding the number of layers $p$ equal, under which condition it has been claimed that DC-QAOA has superior performance than QAOA in quantum optimization for essentially systems of any size  \cite{chandaranaDigitizedcounterdiabaticQuantumApproximate2022,chandaranaDigitizedCounterdiabaticQuantumAlgorithm2023}. However, from Eq.~\eqref{eq:QAOA} and Eq.~\eqref{eq:DC-QAOA} we can see that DC-QAOA adds a cd-term $U_\mathrm{cd}$ to every layer of QAOA. The inclusion of cd-terms in DC-QAOA brings about an inherent trade-off, adding two-qubit gates to the quantum circuit, thereby heightening the circuit complexity even when the number of layers $p$ remains the same as in QAOA. The effect is even more pronounced for NISQ devices due to their decoherence and limited connectivity among qubits.

For a MaxCut problem on graph $\mathcal{G}$ having $M$ edges, the implementation of $U_\mathrm{f}$ requires $2M$ CNOT gates, and the implementation of $U_\mathrm{cd}$ requires $4M$, $2M$ and $0$ CNOT gates for the nested commutator (NC, Fig.~\ref{fig:Fig2}(a)), YY-interaction (YY, Fig.~\ref{fig:Fig2}(b)) and the single qubit Y-rotation  (Y, Fig.~\ref{fig:Fig2}(c)), respectively. Hereafter, we denote the three versions of the DC-QAOA algorithm which uses NC, YY and Y as cd-terms as DC-QAOA(NC), DC-QAOA(YY), and DC-QAOA(Y), respectively. Note that QAOA has $2M$ CNOT gates, therefore DC-QAOA(NC), DC-QAOA(YY), and DC-QAOA(Y) has $6M$, $4M$, and $2M$ CNOT gates in total respectively. 

It is worth noting that DC-QAOA(NC) utilizes three times as many CNOT gates as QAOA, the most among the three versions of the algorithm. Such excessive usage of two-qubit gates in one layer is even more pronounced when we consider large-scale combinatorial optimization. For example, one-layer DC-QAOA(NC) requires $531867$ two-qubit gates while one-layer QAOA requires only $177289$ two-qubit gates based on the topology of Brisbane of IBM Quantum for a complete graph with 127 vertices, which is the number of vertices demonstrated in the most recent experiment  \cite{exp127}. Therefore only assessing the complexity using number of layers is insufficient.

\begin{table}[h!]
  \begin{center}
    \label{tab:table1}
    \begin{tabular}{|l|c|c|} 
      \hline
      Device & QAOA & DC-QAOA(NC) \\
      \hline
        Brisbane  \cite{lalita2024realizing} & 177289 & 531867\\
        Kyiv  \cite{tsunaki2024ensemble} & 177634 & 532902\\
        Torino  \cite{baglio2024data} & 177944 & 533832\\
    \hline
    \end{tabular}
  \end{center}
  \caption{The number of two-qubit gate of two algorithm when compiled into the topology of real quantum devices provided by IBM Quantum using Qiskit transpile function. The underlying graph of MaxCut problem is a complete graph. The native two-qubit gate being compiled to is Echoed Cross-Resonance Gate for Brisbane and Kyiv and Controlled-Z gate for Torino.}
  \label{tab:table1}
\end{table}

The performance of QAOA has been studied extensively in the context of combinatorial optimization problems, with scaling laws providing valuable insights into its efficiency. In  \cite{akshayCircuitDepthScaling2022}, the authors proposed the logistic saturation conjecture in the context of MAX-2-SAT, demonstrating that the critical depth $p_{\mathrm{crit}}$, the minimum QAOA depth required to achieve a specific level of performance, follows a logistic function with respect to the clause density $\alpha$, This logistic model captures how $p_{\mathrm{crit}}$ initially grows rapidly with increasing $\alpha$, before saturating at a depth $p_\mathrm{max}$, beyond which further increases in circuit depth yield diminishing returns. Additionally, their study showed that $p_\mathrm{max}$ scales linearly with the problem size, suggesting that larger instances require proportionally deeper circuits.

In this work, we adapt these ideas to the MaxCut problem, where the analogous metric to clause density is the edge density of the graph. For a graph with $n$ vertices and $m$ edges, the edge density is defined as $\alpha_\mathrm{e}=\frac{2m}{n(n-1)}$, representing the fraction of possible edges that are present. We hypothesize that the logistic saturation conjecture also applies to MaxCut, with $p_\mathrm{crit}$ changing with the edge density of the graph following logistic function and then saturating at a depth $p_\mathrm{max}$. To support this, we provide a formal proof in Appendix~\ref{appendix_a} establishing the equivalence between Max-k-SAT and MaxCut. Moreover, we explore whether $p_\mathrm{max}$ retains a linear relationship with the graph size, providing further evidence for the generality of this scaling behavior across optimization problems. $p_\mathrm{max}$ provides a tool to quantify the advantage afforded by adding cd-terms in DC-QAOA. It is important to see how fast $p_\mathrm{max}$ scales with problem size. For example, if $p_\mathrm{max}$ of DC-QAOA(NC) scales three times slower than $p_\mathrm{max}$ of QAOA, it means that as the problem size grows, DC-QAOA(NC) can eventually solve an optimization problem utilizing fewer CNOT gates than QAOA. Therefore, it is fair to compare DC-QAOA and QAOA by comparing the scaling rate of $p_\mathrm{max}$ versus the problem size. To find the critical depth $p_\mathrm{max}$ for DC-QAOA and QAOA, we first define the critical depth in this study. In general, for a given graph with edge density $\alpha_\mathrm{e}$, the output state of the variational circuit will not cover the whole Hilbert space, the overlap, or fidelity $F$ between the optimal output state of a $p$-layer QAOA and the actual ground state $\ket{\psi_\mathrm{g}}$ is defined as:
\begin{equation}
   F(p,\alpha_\mathrm{e})=\max_{\vec{\alpha},\vec{\beta}}\left|\bra{\psi_\mathrm{g}}U(\vec{\alpha},\vec{\beta})\ket{0}^{\otimes n}\right|^2,
\end{equation}
and for a $p$-layer DC-QAOA, the fidelity is similarly defined as:
\begin{equation}
	F(p,\alpha_\mathrm{e})=\max_{\vec{\alpha},\vec{\beta},\vec{\gamma}}\left|\bra{\psi_\mathrm{g}}U(\vec{\alpha},\vec{\beta},\vec{\gamma})\ket{0}^{\otimes n}\right|^2.
\end{equation}

Fig.~\ref{fig:Fig3} shows the fidelity $F$, calculated at fixed number of layers $p=20$, for a MaxCut problem with 16 vertices (cf. Fig.~\ref{fig:Fig9}(h)) versus number of training iterations for four different versions of the algorithms. Different algorithms converge to different values of $F$, with a higher value indicating a better performance.
Each panel shows eight color lines, which are 8 runs from random initialization of parameters. The converged value of $F$ for different algorithms are averaged and are shown in Table \ref{tab:table1}.
 Overall, the results from DC-QAOA(NC) (Fig.~\ref{fig:Fig3}(b)) is the best, followed by those from DC-QAOA(YY) (Fig.~\ref{fig:Fig3}(c)). Results from DC-QAOA(Y) (Fig.~\ref{fig:Fig3}(d)) and QAOA (Fig.~\ref{fig:Fig3}(a)) are comparable and are both inferior to the other two methods. These results suggest having more CNOT gates on each layer implies  better performance of the algorithm. 


\begin{figure}
    \centering
    \includegraphics[width=\linewidth]{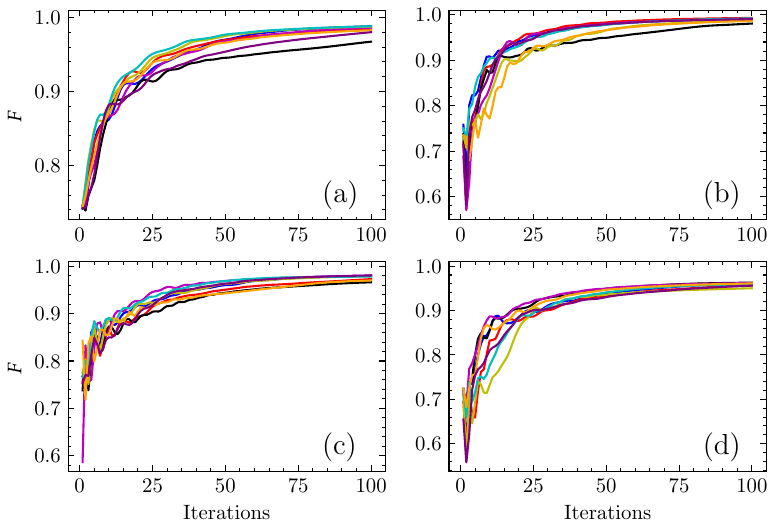}
    \caption{The fidelity $F$, calculated at fixed number of layers $p=20$, for a MaxCut problem versus number of training iterations for four different versions of the algorithms using same number of layers $p=20$. (a) QAOA, (b) DC-QAOA(NC), (c) DC-QAOA(YY) and (d) DC-QAOA(Y). The target graph has 16 vertices. In each panel, the 8 color lines represents 8 random initialization of parameters.}
    \label{fig:Fig3}
\end{figure}

\begin{table}[h!]
  \begin{center}
    \label{tab:table1}
    \begin{tabular}{|l|c|r|} 
      \hline
      Algorithm & Mean & Std. \\
      \hline
        QAOA & 97.85\% & 0.65\%\\
        DC-QAOA(NC) & 98.87\% & 0.38\%\\
        DC-QAOA(YY) & 98.62\% & 0.51\%\\
        DC-QAOA(Y) & 98.30\% & 0.40\%\\
    \hline
    \end{tabular}
  \end{center}
  \caption{The average fidelity $F$ (``Mean'') for four different algorithms, along with their respective standard deviations (``Std.'').}
  \label{tab:table2}
\end{table}

Given some fixed tolerance on performance, $\epsilon>0$, for both algorithm on MaxCut problem defined on graph with edge density $\alpha_\mathrm{e}$, the critical depth $p_\mathrm{crit}$, required to guarantee approximations that fall within an $\epsilon>0$ tolerance, is given by:
\begin{equation}
	p_\mathrm{crit}(\alpha_\mathrm{e})=\min \{p|1-F(p,\alpha_\mathrm{e})<\epsilon\}. \label{eq:critical}
\end{equation}
We can see in Fig.~\ref{fig:new_logistic} that $p_\mathrm{crit}(\rho)$ of both QAOA and DC-QAOA(NC) in MaxCut problem also show an $S$-shaped curve, suggesting that logistic saturation conjecture still holds in this problem. 

To elucidate how the combined effect of the number of CNOT gates and number of layers on the effectiveness of the algorithm, we focus on comparing DC-QAOA(NC), which has the most number of CNOT gates per layer, to QAOA.

Following the procedure in  \cite{akshayCircuitDepthScaling2022}, we fit the $p_\mathrm{crit}(\rho)$ using the following logistic function:
\begin{equation}
	p_\mathrm{crit}(\alpha_\mathrm{e})=\frac{p_\mathrm{max}}{1+e^{k(\alpha_\mathrm{e}-\alpha_\mathrm{e}^{\mathrm{c}})}},
\end{equation}
and in Fig.~\ref{fig:new_scaling}(a) we compare the saturating depth $p_\mathrm{max}$ as functions of the number of qubits $N$, for both QAOA and DC-QAOA(NC). We find that the saturation depth $p_\mathrm{max}$ also scales linearly with the graph size $N$ as shown in the Fig.~\ref{fig:new_scaling}(a). The calculated saturated depths $p_\mathrm{max}$ for both algorithms scale approximately linearly with system size as $p^\mathrm{QAOA}_\mathrm{max} \simeq 3.61N-19.01$ (for QAOA) and $p^\mathrm{DC-QAOA(NC)}_\mathrm{max} \simeq 0.88N-0.51$ (for DC-QAOA(NC)). Therefore, in the limit of solving large scale MaxCut problem on graph $\mathcal{G}$ with $M$ edges, QAOA needs $3.61N$ layers and in turn $3.61N\times2M=7.22NM$ CNOT gates and DC-QAOA(NC) needs $0.82N$ layers and in turn $0.88N\times6M=5.28NM$ CNOT gates. This implies that DC-QAOA(NC) is more advantageous than QAOA in solving large scale optimization problem because the total number of CNOT gates required to achieve the same level of effectiveness is smaller for DC-QAOA(NC), when $N$ is sufficiently large. In other words, the circuit complexity for DC-QAOA(NC) is lower, taking into consideration both the number of effective layers and the number of CNOT gates per layer required.

\begin{figure*}
	\includegraphics[width=\textwidth]{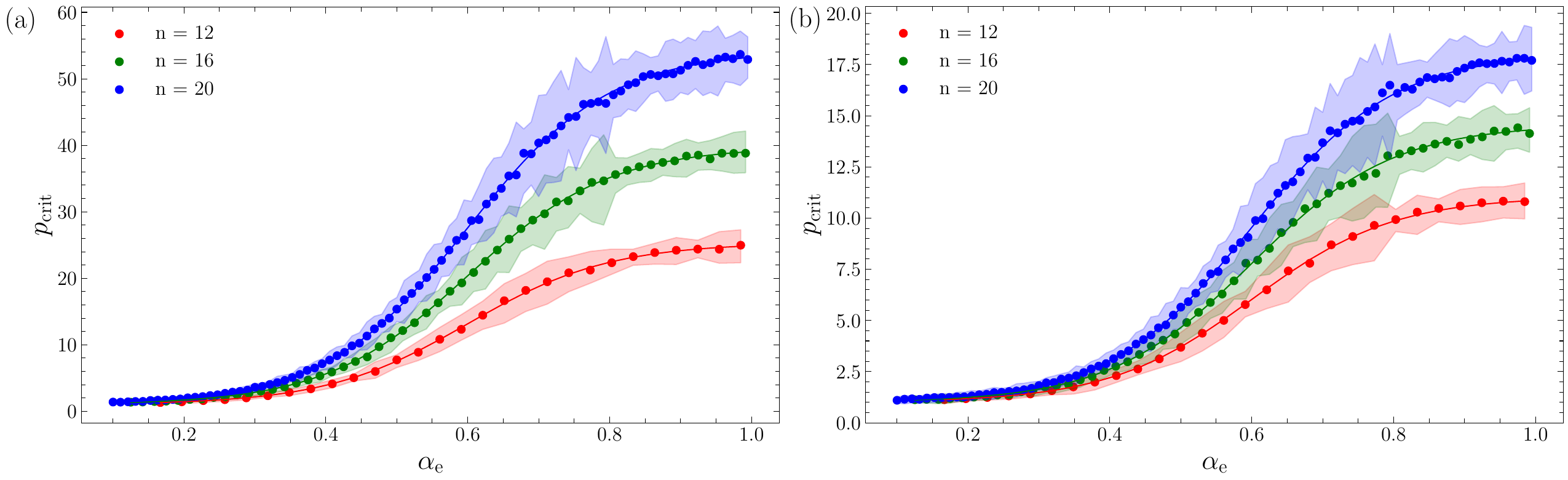}
	\caption{Average critical depth $p_\mathrm{crit}$ against edge density $\alpha_\mathrm{e}$ for QAOA in (a) and DC-QAOA(NC) in (b). Data points represent the average $p_\mathrm{crit}$ calculated according to Eq.~(\ref{eq:critical}) for 50 non-isomorphic graphs for all levels of edge density, and if the number of non-isomorphic graphs is less than 50, we make it up by choosing different initialization of both algorithm. Shaded area represent the standard deviation of $p_\mathrm{crit}$. Different color represent different system sizes.}
    \label{fig:new_logistic}
\end{figure*}

We set $3p^{\mathrm{QAOA}}_\mathrm{max}=p^{\mathrm{DC-QAOA(NC)}}_\mathrm{max}$ and obtain $N \simeq 16.54$. This suggests that DC-QAOA(NC) is more advantageous than QAOA when the number of vertices exceeds 16 (i.e. $N\ge17$). To show this more clearly, we execute DC-QAOA(NC) with $p$ layers, and a QAOA with intentionally set $3p$ layers (which is denoted as QAOA-3 in Fig.~\ref{fig:new_scaling}(b),(c) and (d) and should not be confused with regular QAOA with $p$ layers), to solve the MaxCut problem with different number of vertices. Fig.~\ref{fig:new_scaling}(b) shows the results on 12 vertices (cf. Fig.~\ref{fig:Fig9}(d)), and one can see that QAOA-3 clearly performs better than DC-QAOA(NC).  Fig.~\ref{fig:new_scaling}(c) shows the results on 16 vertices (cf. Fig.~\ref{fig:Fig9}(h)), in which situation the two algorithms are comparable. Given that $N=16$ is still less than 16.94, it is reasonable that QAOA-3 outperforms DC-QAOA(NC) slightly at this stage, although DC-QAOA(NC) rapidly approaches similar performance with increasing $p$. When $p>10$, the difference in $R$ between DC-QAOA(NC) and QAOA becomes negligible. We have also conducted a calculation comparing QAOA-3 and DC-QAOA(NC) for a problem with 20 vertices, with results shown in Fig.~\ref{fig:new_scaling} (d), although the calculation is more expensive. In this case, DC-QAOA(NC) is clearly superior than QAOA-3. Our results shown in Fig.~\ref{fig:new_scaling} confirm that DC-QAOA(NC) is only more advantageous than QAOA on graphs with more than 16 vertices. To our knowledge, this fact has not been noted in the literature and is the key result of this paper.


\begin{figure*}
\includegraphics[width=\textwidth]{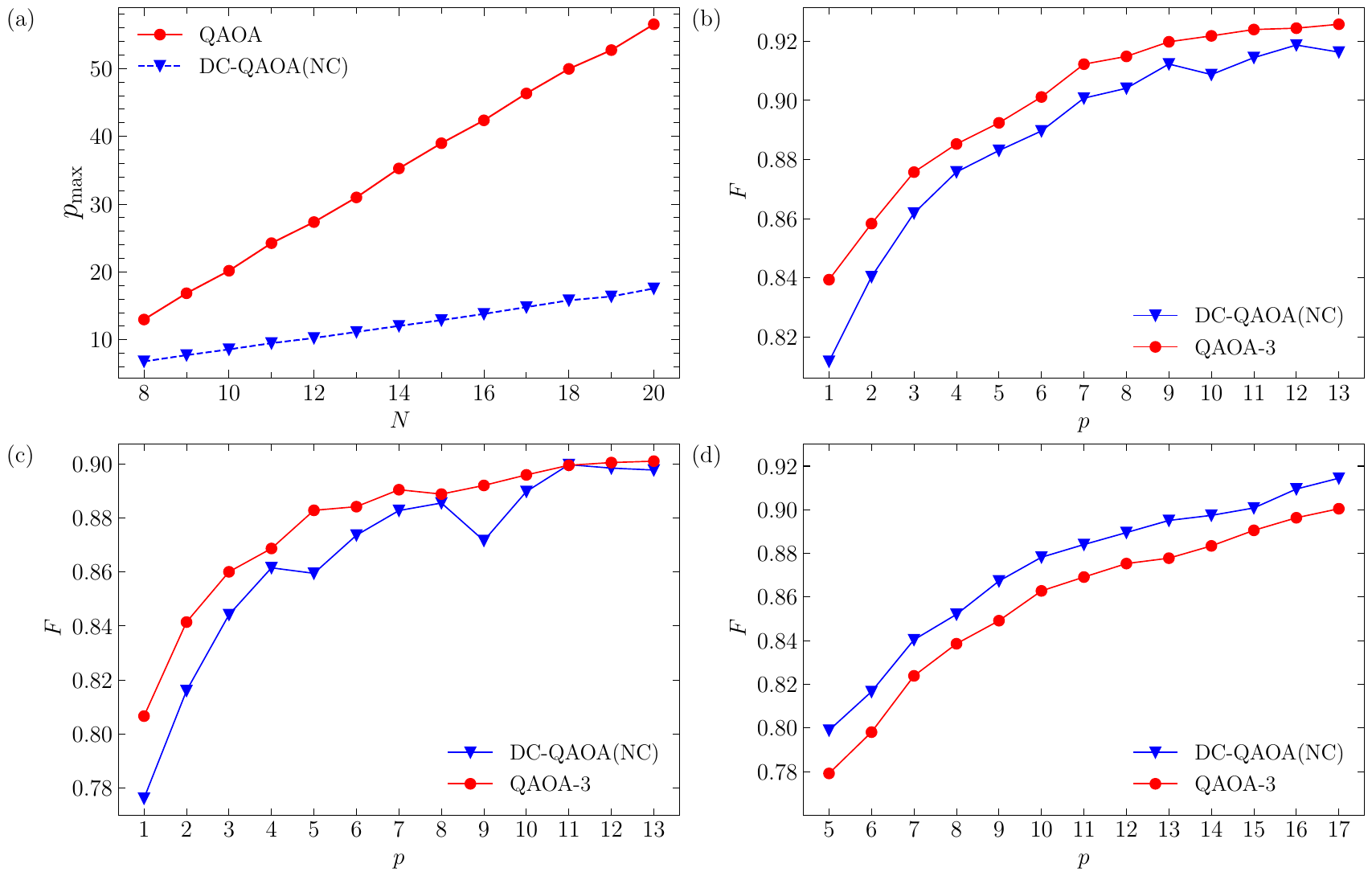}
    \caption{(a) The saturating depth $p_\mathrm{max}$ versus the number of qubits $N$ for DC-QAOA(NC) and QAOA. (b)-(d) The fidelity to ground state $F$ versus $p$ for DC-QAOA(NC) with increasing $p$ layers and for QAOA with $3p$ layers (QAOA-3) using the same number of CNOT gates to find the maximum cut in a graph with  (b) 12, (c) 16  and (d) 20 vertices.}
    \label{fig:new_scaling}
\end{figure*}
\subsection{Feasibility Test in NISQ devices}
In the preceding section, we performed a comparative analysis of the DC-QAOA framework, with particular focus on its DC-QAOA(NC) variant, against the standard QAOA in a noiseless computational setting. This analysis sought to determine whether DC-QAOA(NC), which incorporates the counterdiabatic phenomenon, could surpass the conventional QAOA algorithm in performance as the size of the quantum system scales. A key feature of DC-QAOA(NC) is its reduced reliance on two-qubit gates, such as CNOT gates, which are known to be both error-prone and resource-intensive. By minimizing the usage of these gates, DC-QAOA(NC) presents a promising avenue for achieving greater efficiency and scalability compared to the traditional QAOA framework.

Nevertheless, it is essential to account for the practical limitations imposed by current quantum hardware, particularly in the NISQ era. Quantum circuits with substantial depth are often impractical due to the rapid onset of decoherence and other noise-induced challenges. As the circuit depth increases, errors from gate operations and interactions with the environment accumulate, significantly compromising the fidelity of quantum computations. These constraints necessitate quantum algorithms that balance performance with hardware feasibility. In this context, the practical value of DC-QAOA(NC) lies in its potential to address these limitations by reducing circuit depth, thus improving resilience to noise and decoherence. Examining the trade-offs between algorithmic efficiency and noise tolerance in realistic, hardware-constrained scenarios is vital to fully assessing the advantages of DC-QAOA(NC) over standard QAOA in the NISQ regime.

A widely adopted approach for evaluating the effectiveness of QAOA and its various extensions, both numerically and experimentally, involves benchmarking these algorithms under shallow-circuit conditions, typically with one or a few layers. Such configurations not only reduce the computational and experimental burden but also enable demonstrations on systems extending beyond a small number of qubits. These benchmarks provide critical insights into algorithmic performance and scalability while remaining compatible with the limitations of current quantum devices. Here, we benchmarked one-layer DC-QAOA(NC) and three-layer QAOA, which utilize same number of CNOT gates, on the Sherrington-Kirkpatrick (SK) model  \cite{SKmodel}. We calculate the fidelity between optimal output state of ansatz and the actual ground state,  with 5 to 34 qubits in perfect environment using statevector simulator and with 5 to 17 qubits in noisy environment using quantum trajectory simulation with noise. 

 The SK model describes classical spin system with all-to-all couplings between the $n$ spins and its Hamiltonian is defined as follows:
\begin{equation}
	H_\mathrm{SK}=\frac{1}{2}\sum_{i,j}J_{ij} \sigma_z^{i}\sigma_z^{j},
	\label{eq:SK}
\end{equation}
where $J_{ij}$ is sampled from standard normal distribution. In the noisy environment, each qubit suffers from depolarizing noise when undergoing the CNOT operation and two-qubit depolarization channel can be described by:
\begin{equation}
	\rho'=(1-p)\rho+\frac{pI}{4}.
\end{equation}
In numerical simulation, such channel is applied after every two qubit gate, and by using quantum trajectory simulation method, we simulated this noisy quantum circuit up to 17 qubits. Here we set $p=0.66\%$, which is consistent with the CNOT error rate in the state-of-the-art devices published by IBM Quantum. The result of simulation of 100 randomly generated SK-model for every qubit $N$ can be seen from Fig.~\ref{fig:new_layer}. Both Fig.~\ref{fig:new_layer}(a) and Fig.~\ref{fig:new_layer}(b) show a clear exponential trend with slowly increasing variance. In the perfect environment, the two algorithms reach the same level of performance at around $N=16$ which confirms that one-layer DC-QAOA(NC) has a significantly better exponential scaling compared to three-layer QAOA. Whereas in the noisy environment, the two algorithms reach the same level of performance at around $N=11$. This clearly shows that DC-QAOA(NC) is more robust compared to QAOA. This robustness can be attributed to the fact that counterdiabatic driving terms reduce complexity of the circuit. In Fig.~\ref{fig:new_layer}(c) we can see that DC-QAOA(NC) shows exponential advantage over three-layer standard QAOA, and the advantage is more pronounced in the noisy environment.

\begin{figure*}
    \centering
    \includegraphics[width=\textwidth]{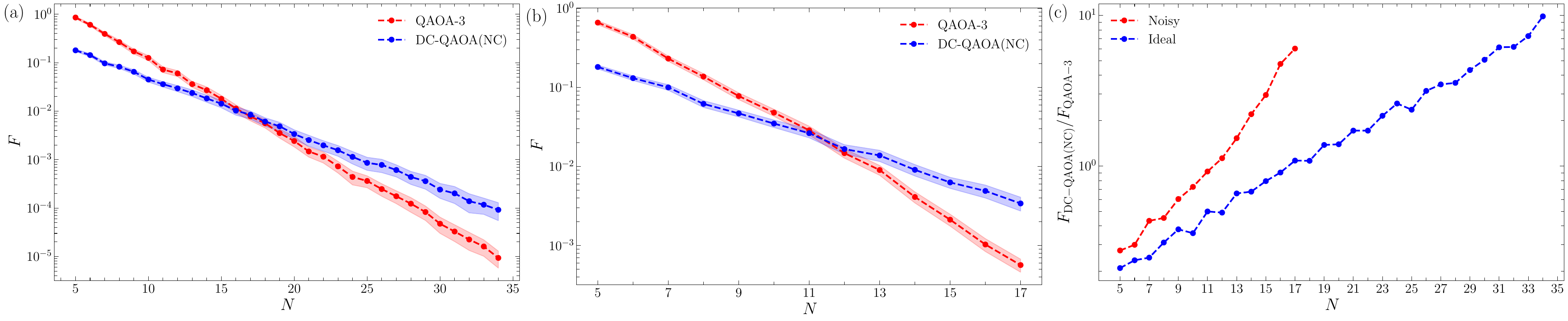}
    \caption{Fidelity $F$ as a function of the number of qubits $N$ in the SK model, defined in Eq.~(\ref{eq:SK}). The comparison is made between a one-layer DC-QAOA(NC) and a three-layer QAOA, both utilizing the same number of CNOT gates. (a) Shows results in an ideal environment without noise, while (b) incorporates depolarization error. One-layer DC-QAOA(NC) performs better than three-layer after $N$ exceeds around 16 in an ideal environment while same happens when $N$ exceeds 11 in the noisy environment. In (c), we shows the exponential improvement of one-layer DC-QAOA(NC) over three-layer QAOA. In ideal environment, the improvement scales as $10^{0.0586N}$ and in noisy environment, the improvement scales as $10^{0.1397N}$}
    \label{fig:new_layer}
\end{figure*}

\subsection{Effective Dimension and Excitation Suppression} \label{suppression}

\begin{figure}
    \centering
    \includegraphics[width=\linewidth]{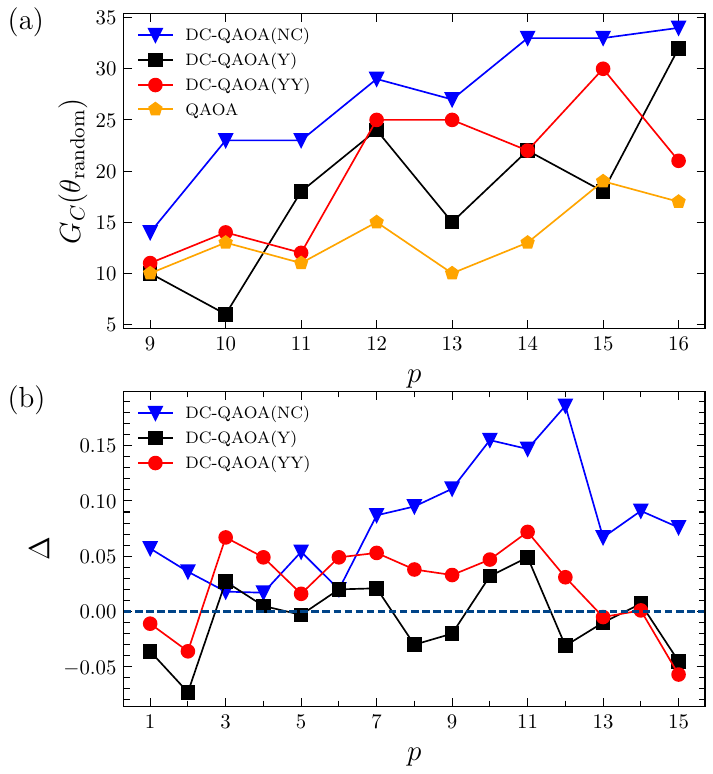}
    \caption{(a) $G_C(\theta_\mathrm{random})$ as functions of $p$ for DC-QAOA with three different kind of cd-terms, as well as QAOA. (b) The suppression ability $\Delta$ as functions of $p$ for the three types of DC-QAOA.}
    \label{fig:Fig5}
\end{figure}

In this section we delve deeper into the question why certain versions of the algorithm is more efficient than others. As we shall show in the following, the answer lies in the number of effective dimensions of the parameter space  \cite{haugCapacityQuantumGeometry2021a}. 
QAOA is, essentially,  a parameterized version of trotterized adiabatic quantum evolution: for  QAOA with $p$ layers, the corresponding adiabatic evolution is parameterized with $2p$ parameters $(\alpha_1,\alpha_2,...,\alpha_p;\beta_1,\beta_2,...,\beta_p)$. In other words, the parameter space in which one may search for the optimal realization of the algorithm has $2p$ dimensions. A larger $p$ implies a parameter space with greater dimensions, and the resulting QAOA algorithm is more powerful. On the other hand, DC-QAOA  increases the dimensions of the parameter space by adding cd-terms $U_\mathrm{cd}(\gamma)$ to every layer of the original QAOA, resulting in $3p$ parameters $(\alpha_1,\alpha_2,...,\alpha_p;\beta_1,\beta_2,...,\beta_p;\gamma_1,\gamma_2,...,\gamma_p)$. Despite the different physical intuitions behind QAOA and DC-QAOA, both algorithms attempt to increase the dimensions of the parameter space by increasing the complexity of quantum circuits. 

However, not all dimensions of the parameter space are independent. The ``effective'' dimensions, or independent dimensions, are important to parameterized quantum circuits  \cite{haugCapacityQuantumGeometry2021a}. 
The independent quantum dimensions, denoted as $D_C$ in  \cite{haugCapacityQuantumGeometry2021a}, refers to the number of independent parameters that the parameterized quantum circuit can express in the space of quantum states, i.e., in the total $2^N-1$ parameters (where the $-1$ is a result of normalization) for a generic $N$-qubit quantum state. Furthermore,  it can be calculated by determining $G_C(\theta)$ which counts the number of independent directions in the state space that can be accessed by an infinitesimal update of $\theta$, where $\theta$ includes all the trainable parameters in a parameterized quantum circuit. $G_C(\theta)$ can be evaluated by counting the number of non-zero eigenvalues of Quantum Fisher Information matrix $F(\theta)$ \cite{haugCapacityQuantumGeometry2021a}:
\begin{equation}
    G_C(\theta) = \sum_i^M I[\lambda^{(i)}(\theta)],
\end{equation}
where $I(x)=0$ for $x=0$ and $I(x)=1$ for $x\neq 0$. $F(\theta)$ is defined as:
\begin{equation}
    F_{ij}(\theta) = \mathrm{Re}(\braket{\partial_i\psi_\mathrm{o}|\partial_j\psi_\mathrm{o}}-\braket{\partial_i\psi_\mathrm{o}|\psi_\mathrm{o}}\braket{\psi_\mathrm{o}|\partial_j\psi_\mathrm{o}}),
\end{equation}
where $\ket{\psi_\mathrm{o}}$ is the output state of quantum circuit and $\ket{\partial_i\psi_o}$ is the change in output state upon an infinitesimal change in the $i$-th element of all the trainable parameters $\theta$.
It is worth to note that $G_C(\theta)$ can faithfully represent the number of independent parameters $D_C$ only when $\theta$ is randomly sampled from all possible parameters, under the condition that the parameterized quantum circuit is layer-wise with parameterized gates being rotations resulting from Pauli operators which enjoys 2$\pi$ periodicity in angles.

In Fig.~\ref{fig:Fig5}(a), we show $G_C(\theta_\mathrm{random})$ as functions of $p$ for QAOA and all three versions of DC-QAOA. We can see that for $p>10$, DC-QAOA always has more effective dimensions than QAOA, while the effective dimension for DC-QAOA(NC) is most pronounced. 
The cd-terms in DC-QAOA are derived from analog cd-driving, introduced to maintain the system in the ground state by preventing unwanted excitations. Extra dimensions provided by adding cd-terms can suppress the non-diagonal probabilities in the transition matrix, which excites the ground state of $H_\mathrm{i}$ to the excitation states of $H_\mathrm{f}$ and boost the performance of DC-QAOA. 
It is therefore interesting to examine the ability of different types of cd-terms to suppress unwanted excitations. To do this comparison, we first conduct the DC-QAOA as usual, resulting in the original approximate ratio $R_\mathrm{raw}$. Then, we conduct DC-QAOA again but with all cd-terms removed (all other parameters unchanged), arriving at a new approximate ratio $R_\mathrm{nocd}$. The suppression ability of this type of cd-terms, denoted as $\Delta$, is defined as the difference between $R_\mathrm{nocd}$ and the original approximate ratio $R_\mathrm{raw}$:
\begin{equation}
    \Delta = R_\mathrm{raw}-R_\mathrm{nocd}.
\end{equation}.

 In Fig.~\ref{fig:Fig5}(b) we show the comparison of the suppression ability $\Delta$ as functions of $p$ for the three types of DC-QAOA with 15 layers trained on a graph with 16 vertices. A positive $\Delta$ indicates that the cd-terms are effective in increasing the approximate ratio, and in turn, the efficiency of the algorithm. We see that for the case of DC-QAOA(NC), $\Delta$ is positive for the entire range of $p$, while for other methods  $\Delta$ fluctuates between positive and negative values. If we focus on the range $6<p<15$,  DC-QAOA(NC) has  $\Delta$ values higher than any other methods. This fact indicates the strong ability of DC-QAOA(NC) to suppress unwanted excitation and to add the effective dimensions, and thus its superior performance. On the other hand, the suppression ability for DC-QAOA(Y) and DC-QAOA(YY) are only positive for intermediate $p$ values, limiting its efficiency. This also explains why DC-QAOA(Y) and DC-QAOA(YY) do not exhibit considerable advantage over QAOA-3.
 
 
\subsection{Instance-Sequential Training (IST)} \label{training}

In previous sections, we have shown that DC-QAOA can be more advantageous than QAOA under certain conditions, because the former has accelerated critical depth scaling behavior and more effective dimensions to suppress non-diagonal transitions. In this section, we study the distribution of optimal parameters resulting from a judiciously trained DC-QAOA(NC). We shall show that these parameters are more concentrated in the parameter space, and based on this fact, we propose an IST method, which can substantially improve the efficiency in training.


For a $p$-layer QAOA defined on $N$ qubits, the optimal parameters $(\boldsymbol{\alpha}_N,\boldsymbol{\beta}_N)$ are said to be concentrated if:
\begin{equation}
    \begin{split}
      &\exists l>0: \forall \boldsymbol{\alpha}_N, \boldsymbol{\beta}_N, \exists \boldsymbol{\alpha}_{N+1}, \boldsymbol{\beta}_{N+1}: \\ 
      &d_N \equiv |\boldsymbol{\alpha}_N-\boldsymbol{\alpha}_{N+1}|^2+|\boldsymbol{\beta}_N-\boldsymbol{\beta}_{N+1}|^2 = \mathcal{O}\left(\frac{1}{N^l}\right)
    \end{split},
\end{equation}
where $(\boldsymbol{\alpha}_N,\boldsymbol{\beta}_N)$ may not be the unique optimal parameters,  $l \geq 2$ can guarantee that the optimal parameters will converge to a limit when $N \rightarrow \infty$  \cite{akshayParameterConcentrationQuantum2021}, and $d_N$ is the distance between optimal parameters found for two graphs whose number of vertices differ by one.  For DC-QAOA(NC), parameter concentration can be similarly defined on $(\boldsymbol{\alpha}_N,\boldsymbol{\beta}_N,\boldsymbol{\gamma}_N)$.

The dimension of the parameter space is typically high, and to visualize it we use the t-SNE algorithm  \cite{van2008visualizinga}. This algorithm is ubiquitously used for dimensionality reduction  while maintaining the local structure of the data, making clusters in the original high-dimensional space more apparent and easily identifiable in the reduced-dimensional space  \cite{van2008visualizinga}. 

We focus on the set of graphs as shown in Fig.~\ref{fig:Fig9}. For each graph, we optimize 10 sets of parameters of DC-QAOA(NC) and QAOA respectively from random initial parameters, in a 10-layer DC-QAOA(NC) and QAOA respectively. In Fig.~\ref{fig:Fig6}, all optimal parameters are embedded into two dimensions to visualize by t-SNE algorithm. For DC-QAOA(NC) in Fig.~\ref{fig:Fig6}(a), at least one set of optimal parameter of all graph size appears within the red box and roughly forms a curly shaped trajectory, whereas for QAOA in Fig.~\ref{fig:Fig6}(b), optimal parameters are more scattered, only creating small-scale clusters, if any, in different regions. 

\begin{figure}
    \centering
    \includegraphics[width=\linewidth]{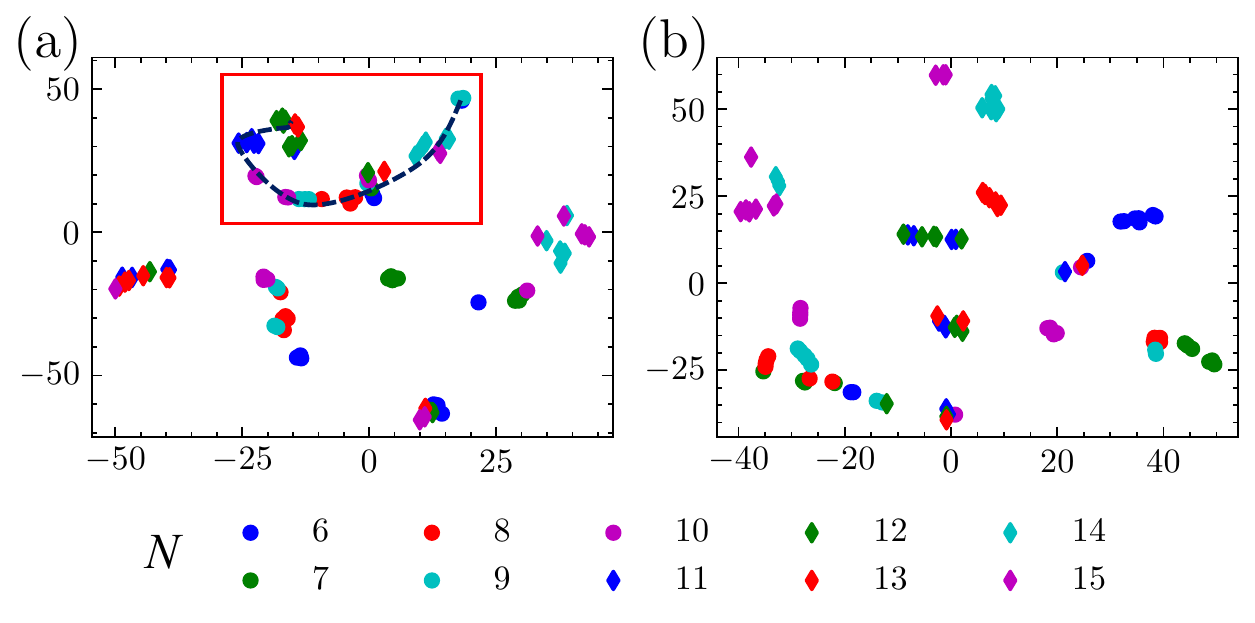}
    \caption{10 Optimal parameters for graphs with vertices ranging from 6 to 15 are embedded into a two dimensional manifold using t-SNE for (a) DC-QAOA(NC) and (b) QAOA. Graph with $n$ vertices is transformed into graph with $n-1$ vertices by randomly removing one vertex. In (a), optimal parameters of DC-QAOA(NC) for graphs with different number of vertices can all be found in the red box whereas in (b), optimal parameters of QAOA are more scattered.}
    \label{fig:Fig6}
\end{figure}

To examine parameter concentration in a more quantitative fashion, the distance $d_N$ between $(\boldsymbol{\alpha}_N,\boldsymbol{\beta}_N)$ and $(\boldsymbol{\alpha}_{N+1},\boldsymbol{\beta}_{N+1})$ for QAOA and the distance between $(\boldsymbol{\alpha}_N,\boldsymbol{\beta}_N,\boldsymbol{\gamma}_N)$ and $(\boldsymbol{\alpha}_{N+1},\boldsymbol{\beta}_{N+1},\boldsymbol{\gamma}_{N+1})$ for DC-QAOA(NC) are plotted against  the logarithm of number of qubits $N$  in Fig.~\ref{fig:Fig7}. Distance is calculated between two closest ones from 100 sets of optimal parameters between graphs with $N$ qubits and $N+1$ qubits. For DC-QAOA(NC), the logarithm of the distance shows a linear trend with number of qubits, indicating that there is a consistent trend of parameter concentration in DC-QAOA(NC) with $l \simeq 7 $, implying that optimal parameters of DC-QAOA will converge to a limit when $N$ is large enough.

To fully leverage the advantage afforded by such phenomenon of parameter concentration  in DC-QAOA(NC), we propose the IST method, that can make the training process of DC-QAOA(NC) more efficient in using quantum resources. Under the IST method, for a graph with $N$ vertices, we randomly remove $N-k$ vertices (and the edges associated with them) while keeping the graph connected, i.e. the graph should not be separated to two. Firstly, we train DC-QAOA(NC) on the remaining $k$ vertices to get optimal parameters $(\boldsymbol{\alpha}_k,\boldsymbol{\beta}_k,\boldsymbol{\gamma}_k)$. Then, one vertex is added back to the graph (and the edges associated with it)  and DC-QAOA(NC) is trained again using $(\boldsymbol{\alpha}_k,\boldsymbol{\beta}_k,\boldsymbol{\gamma}_k)$ as the starting point to obtain $(\boldsymbol{\alpha}_{k+1},\boldsymbol{\beta}_{k+1},\boldsymbol{\gamma}_{k+1})$. We then add one more vertex back, and the above steps are repeated until $(\boldsymbol{\alpha}_N,\boldsymbol{\beta}_N,\boldsymbol{\gamma}_N)$ is obtained. The detailed procedure of the IST method is outlined in Algorithm \ref{alg:IST}, and an illustrative flow chart is presented in Fig.~\ref{fig:Flow_chart}.

\begin{algorithm}
\caption{Instance-Sequential Training (IST) Method}
\label{alg:IST}
\begin{algorithmic}[1]
\REQUIRE Graph $ G $ with $ N $ vertices
\STATE Initialize an empty set of parameters $ \Theta $
\STATE Randomly remove $ N-k $ vertices to form a reduced graph $ G_k $, ensuring connectivity
\STATE Train the DC-QAOA circuit on $ G_k $ to find optimal parameters $ \Theta_k $
\STATE Set $ \Theta = \Theta_k $
\FOR{$ i = k+1 $ to $ N $}
    \STATE Add one vertex back to $ G $ to form $ G_i $
    \STATE Train DC-QAOA on $ G_i $ using $ \Theta $ as the initial parameters
    \STATE Update $ \Theta $ with the new optimal parameters
\ENDFOR
\RETURN Final parameters $ \Theta $
\end{algorithmic}
\end{algorithm}

\begin{figure*}
	\includegraphics[width=\textwidth]{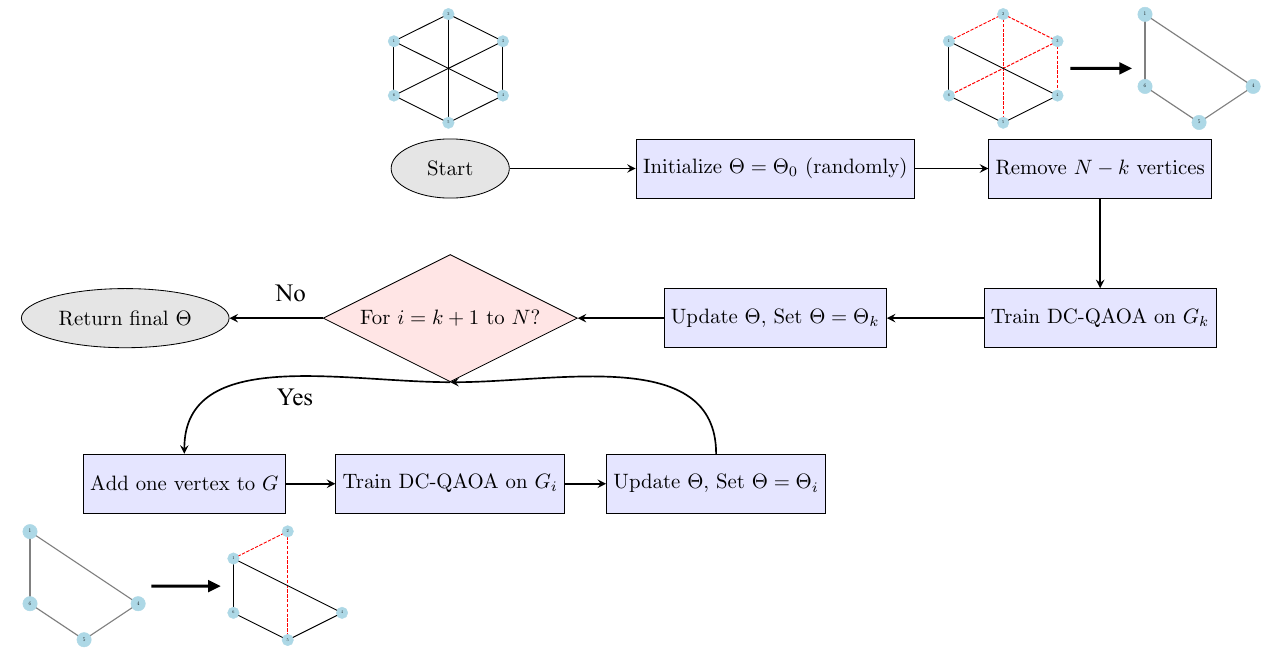}
	\caption{An illustrative flow chart of IST algorithm. In the illustration, we use a six-node graph as an example. The original graph is shown above "start", after initialization of $\Theta$ with $\Theta_0$, two nodes are removed from the graph, together with all the edges connected to these nodes, shown by red dashed line. In the loop process, one node is added back to the graph to conduct another round of training, together with its edges with all the existing nodes, also shown by red dashed line.}
	\label{fig:Flow_chart}
\end{figure*}

On the other hand, traditional training method needs to construct a circuit on $N$ qubits based on the connected edges in the graph. The number of CNOT gates in $U_\mathrm{f}(\beta)$ and $U_\mathrm{cd}(\gamma)$ are the same as the number of edges in the graph, which, for large problems necessarily complicating the circuit, makes it more prone to decoherence. In order to quantify and compare the quantum resources used in the two training methods, we use the product of number of qubits $N$, number of CNOT gates (equivalent to $6M$ for DC-QAOA(NC) and $2M$ for QAOA respectively, where $M$ is the number of edges in the graph $\mathcal{G}$) and number of training iterations $c$ as an indicator. The proposed indicator for comparing quantum resources captures three critical dimensions of resource usage in quantum training:
\begin{itemize}
    \item Number of qubits $N$ represents the system size, directly influencing the computational complexity and the memory resources required to simulate or run the quantum system.
    \item Number of CNOT gates ($\propto M$) reflects the entangling operations, which are among the most resource-intensive quantum operations due to their role in creating quantum correlations.
    \item Training iterations $c$ represent the temporal aspect of resource consumption, as more iterations correspond to greater computation time.
\end{itemize}
Together, this indicator effectively quantifies the overall computational burden by incorporating both spatial (qubits and CNOT gates) and temporal (iterations) aspects, giving a holistic view of the quantum resources consumed during the training process.
The quantum resources used in the two training methods are summarized in Table~\ref{tab:table2}. It is clear that $\sum_{i=k}^N \frac{c}{N-k+1}\times6M_i\times i < \sum_{i=k}^N \frac{c}{N-k+1}\times 6M\times i=6cM\frac{N+k}{2}$, which means that IST uses at most half of quantum resources used by traditional method when $k$ is sufficiently small compared to $N$ as in the case of large-scale quantum optimization where $N$ is large.

\begin{table}[h!]
  \begin{center}
    \label{tab:table2}
    \begin{tabular}{|l|r|} 
      \hline
      Method & Quantum Resources \\
      \hline
        IST & $\sum_{i=k}^N \frac{c}{N-k+1}\times6M_i\times i$  \\
      \hline
        Traditional & $6cMN$ \\
    \hline
    \end{tabular}
  \end{center}
  \caption{The quantum resources used by training DC-QAOA(NC) via IST and traditional training methods respectively. $M_i$ is the number of edges when there are $i$ remaining vertices in the graph.}
  \label{tab:table2}
\end{table}

To demonstrate the effectiveness of the IST method, we compare it to the traditional training method while fixing the number of total training iterations the same. In this comparison, the IST method uses much fewer quantum resources because it does not have to access all $N$ qubits in the quantum processing unit in every iteration. We benchmark the advantage of performance by the proportion of approximate ratio $R_\mathrm{IST}$ of the IST method relative to the one for traditional method $R_\mathrm{traditional}$. In Fig.~\ref{fig:Fig8}, the performance of DC-QAOA(NC) trained using IST method is nearly the same with the one trained using traditional method using much less quantum resources. Sometimes, for $N=9$ and 13 DC-QAOA(NC) trained using IST works even better than the one trained using traditional method. However, for QAOA, the IST method fails due to the lack of parameter concentration. We conclude that when the number of training iterations are 
held the same for both methods, IST method, on average, takes less time to run a quantum device and uses fewer qubits to construct the circuits. This is a favorable characteristic especially when the access to the device is limited by certain quota, e.g. cloud-based ones. 

\begin{figure}[t]
    \centering
    \includegraphics[width=\linewidth]{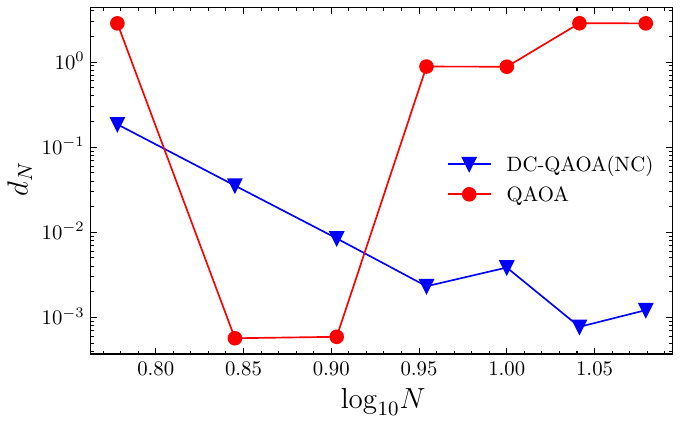}
    \caption{Logarithm of distance between optimal parameters versus qubits (vertices) for different graphs shown in Fig.~\ref{fig:Fig9}. Distance is calculated based on the closest optimal parameters for graphs with $N$ vertices and $N+1$ vertices among 100 random initializations. For DC-QAOA(NC), logarithm of distance shows a linear trend indicating that the optimal parameters are much more concentrated than in QAOA.}
    \label{fig:Fig7}
\end{figure}

\begin{figure}[t]
    \centering
    \includegraphics[width=\linewidth]{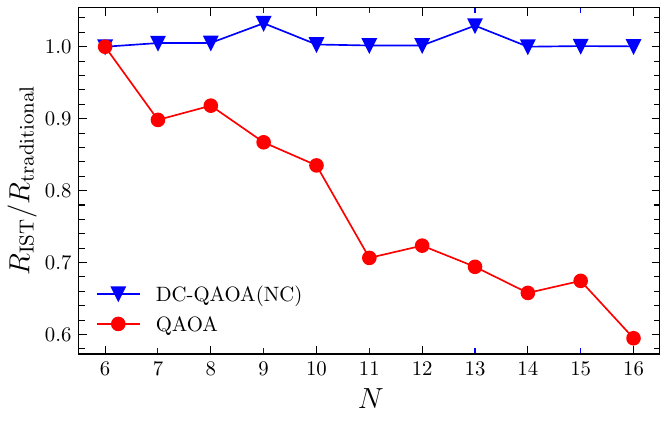}
    \caption{A comparative analysis between IST and traditional training methods in DC-QAOA(NC) and QAOA. IST not only matches but sometimes surpasses traditional methods in achieving optimal expectation values, demonstrating superior efficiency with significantly fewer quantum resources because parameter concentration.}
    \label{fig:Fig8}
\end{figure}

\section{CONCLUSION and Outlook} \label{conclude}

In this work, we have used the MaxCut problem as an example to thoroughly examine the comparison between DC-QAOA and QAOA. Such a comparison has been previously done with the number of layers held equal and with a conclusion that DC-QAOA is superior to QAOA for all sizes of the problem. However, we demonstrate that this statement must be carefully refined according to the actual circuit complexity afforded by different algorithms, which includes both the number of layers and number of CNOT gates in each layer. The cd-terms introduced in DC-QAOA reduces the critical depth of the algorithm but also introduces additional CNOT gates, so the overall circuit complexity is not necessarily lower than QAOA without cd-terms. We have found that for MaxCut problem with $N<16$, QAOA is more advantageous than DC-QAOA(NC), while for larger problems with $N>16$, DC-QAOA(NC) is superior. This refinement of the comparison between the two algorithms has not been previously noted in the literature. Although both YY and Y types of cd-terms enhance the performance of QAOA, the magnitude of improvement is not as significant as that achieved by DC-QAOA(NC).

To further validate the practicality of DC-QAOA, we benchmark one-layer DC-QAOA(NC) against three-layer QAOA, ensuring they use the same number of CNOT gates. In noiseless simulations, we find that DC-QAOA(NC) starts outperforming QAOA when the problem size exceeds $N=16$. More importantly, in a noisy environment simulated using a depolarization model with realistic error rates in NISQ devices, DC-QAOA(NC) achieves an advantage at an even smaller problem size ($N \approx 11$). This result highlights the robustness of DC-QAOA(NC) against noise, making it a promising approach for near-term quantum hardware.

Furthermore, we calculated the effective dimensions of the algorithms and have shown that the superior performance of DC-QAOA(NC) in larger problems is rooted in its strong ability to suppress unwanted excitation and to add the effective dimensions. We have also noted that the optimal parameters found by DC-QAOA(NC) are more concentrated than those from QAOA.
Based on this finding,  we introduced an IST training method, which achieves similar (and sometimes even better) results as compared to traditional training methods, while consuming much fewer quantum resources. Our results should be valuable to recent efforts in utilizing VQA to solve combinatorial optimization problems for large-scale computational tasks, thereby advancing the study of quantum optimization.

While our study focuses on comparing DC-QAOA with standard QAOA, it is important to acknowledge that numerous other QAOA variants have been explored, each designed to enhance specific aspects of the algorithm. Broadly, these variants can be classified into three main categories. The first category aims to increase the expressivity and flexibility of the ansatz by introducing more variational parameters, as seen in Expressive QAOA (XQAOA) \cite{XQAOA} and Multi-Angle QAOA (MA-QAOA) \cite{multi1,multi2}. Since standard QAOA can be viewed as a subset of these approaches, they ideally perform at least as well as standard QAOA in the worst case. The same motivation can be extended to DC-QAOA, as adding more parameters could potentially enhance its optimization capability. However, this comes at the cost of increased sampling overhead, as estimating gradients for a higher number of parameters demands more measurements, exposing the algorithm to greater shot noise. This issue is particularly significant in the NISQ era, where hardware constraints limit measurement precision. Furthermore, overparameterization alters the loss landscape of the cost function, potentially leading to trainability issues such as barren plateaus. Future studies should carefully examine whether adding more parameters to DC-QAOA (particularly DC-QAOA(NC)) yields a genuine performance improvement or simply exacerbates these challenges.

The second category of QAOA variants focuses on redesigning the mixing operator instead of relying on the problem Hamiltonian $H_f$, as in the case of Grover Mixer QAOA (GM-QAOA) \cite{grover1,grover2}. This poses a fundamental challenge for DC-QAOA because it originates from counterdiabatic driving, which depends on constructing an appropriate gauge potential to suppress excitations. When a different mixing Hamiltonian is introduced, designing a corresponding counterdiabatic term becomes nontrivial, potentially requiring new approaches to ensure diabatic transitions are effectively minimized. Exploring how counterdiabatic strategies can be adapted to non-standard mixing operators presents an interesting direction for future research.

A third category of variants seeks to reduce the quantum computational burden by leveraging classical preprocessing techniques that exploit problem structure, as seen in Recursive QAOA (RQAOA) \cite{RQAOA} and ST-Spanning Tree QAOA (ST-QAOA) \cite{STQAOA}. These methods apply classically solvable criteria to simplify the problem before executing the quantum algorithm. Unlike the previous categories, these techniques can be seamlessly integrated into DC-QAOA without fundamentally altering its counterdiabatic framework. Future work could explore how such classical-quantum hybrid strategies can further enhance the practical implementation of DC-QAOA, particularly in scaling to larger problem instances while mitigating quantum resource constraints.

Overall, our findings suggest that DC-QAOA provides a meaningful alternative to standard QAOA, particularly when considering circuit complexity in the NISQ era. However, further research is needed to explore its applicability across broader problem classes, its potential integration with other QAOA variants, and its behavior under realistic quantum hardware noise. A deeper understanding of its trainability, parameter concentration properties, and adaptation to alternative ansatz will be essential in determining its viability for large-scale quantum optimization in the near future.

\section*{ACKNOWLEDGEMENT}

This work is supported by the National Natural Science Foundation of China (Grant Nos.~11874312 and 12474489), the Research Grants Council of Hong Kong (CityU 11304920), Shenzhen Fundamental Research Program (Grant No. JCYJ20240813153139050), the Guangdong Provincial Quantum Science Strategic Initiative (Grant No. GDZX2203001, GDZX2403001), and the Innovation Program for Quantum Science and Technology (Grant No. 2021ZD0302300).

\appendix

\section{Proof of equivalence between Max-k-SAT and MaxCut} \label{appendix_a}

Max-k-SAT and MaxCut are two fundamental problems in combinatorial optimization and computational complexity. Both are known to be NP-hard and widely studied in theoretical computer science. In this appendix, we present a detailed polynomial-time reduction between Max-k-SAT and MaxCut in both directions, thereby demonstrating their computational equivalence. While such reductions are standard (see, e.g., \cite{reduction1,reduction2,reduction3}), we include them here for completeness and pedagogical clarity, particularly to motivate why the logistic saturation conjecture, initially formulated for Max-2-SAT, can be extended to MaxCut.

\subsubsection{Max-k-SAT}

Let $ \varphi $ be a Boolean formula in \textbf{Conjunctive Normal Form (CNF)}, which consists of $ m $ clauses, each containing at most $ k $ literals. A \textbf{literal} is a variable $ x_i $ or its negation $ \neg x_i $, where $ x_i \in \{x_1, x_2, \dots, x_n\} $ is a Boolean variable.

We define the \textbf{Max-k-SAT} problem as follows:

\begin{enumerate}
    \item \textbf{Input}: 
    \begin{itemize}
    	\item A set of Boolean variables $ \{x_1, x_2, \dots, x_n\} $.
        \item A Boolean formula $ \varphi = C_1 \land C_2 \land \cdots \land C_m $, where each clause $ C_i $ is a disjunction of at most $ k $ literals, i.e., 
        \[
        C_i = \ell_{i1} \vee \ell_{i2} \vee \cdots \vee \ell_{ik}
        \]
        where each $ \ell_{ij} $ is a literal which is either a variable or its negation (for example, $ x $ or $ \neg x $). In this way, $ \ell_{ij} \in \{x_1, x_2, \dots, x_n, \neg x_1, \neg x_2, \dots, \neg x_n\} $.
    \end{itemize}
    
    \item \textbf{Goal}:  
    Find a truth assignment $ \mathbf{A} = (a_1, a_2, \dots, a_n) $ where $ a_i \in \{\text{True}, \text{False}\} $ for every Boolean variable, such that the \textbf{number of satisfied clauses} is maximized.

    \item \textbf{Satisfaction of a Clause}:  
    A clause $ C_i = \ell_{i1} \vee \ell_{i2} \vee \cdots \vee \ell_{ik} $ is said to be \textbf{satisfied} if at least one of its literals evaluates to true under the truth assignment $ \mathbf{A} $. Formally, for a clause $ C_i $, it is satisfied if there exists at least one $ j \in \{1, 2, \dots, k\} $ such that:
    \[
    \text{Evaluate}(\ell_{ij}, \mathbf{A}) = \text{True}
    \]
    where $ \text{Evaluate}(\ell_{ij}, \mathbf{A}) = \text{True} $ if $ \ell_{ij} $ evaluates to true under the assignment $ \mathbf{A} $, and $ \text{False} $ otherwise.

    \item \textbf{Objective}:  
    Maximize the number of satisfied clauses, i.e., the goal is to find the assignment $ \mathbf{A} $ such that the number of satisfied clauses $ \text{Satisfy}(\mathbf{A}) $, defined as:
    \[
    \text{Satisfy}(\mathbf{A}) = \left| \left\{ i \mid C_i \text{ is satisfied by } \mathbf{A} \right\} \right|
    \]
    is maximized.

    \item \textbf{Maximization}:  
    The problem is to find the truth assignment $ \mathbf{A} $ that maximizes $ \text{Satisfy}(\mathbf{A}) $. Mathematically:
    \[
    \mathbf{A^*} = \arg \max_{\mathbf{A}} \text{Satisfy}(\mathbf{A})
    \]
    where $ \mathbf{A^*} $ is the truth assignment that maximizes the number of satisfied clauses.
\end{enumerate}

In summary, Max-k-SAT is a \textit{maximization problem} where the aim is to optimize the truth assignment to variables such that the total number of clauses that evaluate to true is as large as possible. One example of Max-3-SAT is shown in TABLE~\ref{tab:max3sat} where for each of the \(2^3 = 8\) possible truth assignments, we indicate which clauses are satisfied and count the total.
The goal in Max-3-SAT is to find the assignment that satisfies the most clauses. In this example, 5 assignments satisfy all three clauses, achieving the maximum value.

\begin{table}[h!]
\centering

\begin{tabular}{|c|c|c|c|c|c|c|}
\hline
\(x_1\) & \(x_2\) & \(x_3\) & Clause 1 & Clause 2 & Clause 3 & \# Satisfied \\
\hline
T & T & T & Satisfied & Satisfied & Satisfied & \textbf{3} \\
T & T & F & Satisfied & Satisfied & Satisfied & \textbf{3} \\
T & F & T & Satisfied & Not Satisfied & Satisfied & 2 \\
T & F & F & Satisfied & Satisfied & Satisfied & \textbf{3} \\
F & T & T & Satisfied & Satisfied & Satisfied & \textbf{3} \\
F & T & F & Satisfied & Satisfied & Not Satisfied & 2 \\
F & F & T & Satisfied & Satisfied & Satisfied & \textbf{3} \\
F & F & F & Not Satisfied & Satisfied & Satisfied & 2 \\
\hline
\end{tabular}

\caption{Truth table for a Max-3-SAT instance with clauses \(C_1 = x_1 \lor x_2 \lor x_3\), \(C_2 = \lnot x_1 \lor x_2 \lor \lnot x_3\), and \(C_3 = x_1 \lor \lnot x_2 \lor x_3\). The final column shows how many clauses are satisfied by each assignment. The maximum is highlighted.}
\label{tab:max3sat}
\end{table}

\subsubsection{Max-Cut}

The \textbf{Max-Cut} problem is a well-known optimization problem in graph theory. Given a graph, the goal is to partition the vertices of the graph into two disjoint sets such that the number of the edges crossing the partition is maximized.

We define the \textbf{Max-Cut} problem as follows:

\begin{enumerate}
    \item \textbf{Input}: 
    \begin{itemize}
        \item A graph \( G = (V, E) \), where:
        \begin{itemize}
            \item \( V = \{v_1, v_2, \dots, v_n\} \) is the set of vertices, with \( n \) vertices.
            \item \( E = \{e_1, e_2, \dots, e_m\} \) is the set of edges, where each edge \( e_i = (v_a, v_b) \) connects two vertices \( v_a, v_b \in V \).
        \end{itemize}
        \item A partition of the vertices into two sets \( V_1 \subseteq V \) and \( V_2 = V \setminus V_1 \), where \( V_2 \) is the complement of \( V_1 \), meaning \( V_1 \cup V_2 = V \) and \( V_1 \cap V_2 = \emptyset \).
    \end{itemize}

    \item \textbf{Goal}:  
    Find a partition \( (V_1, V_2) \) of the vertex set \( V \) such that the number of edges between \( V_1 \) and \( V_2 \) is maximized. That is, we want to maximize:
    \[
    \text{Cut}(V_1, V_2) = \left| \left\{ e \in E \mid e \text{ connects a vertex in } V_1 \text{ to a vertex in } V_2 \right\} \right|
    \]
    where the cut size is the number of edges that have one endpoint in \( V_1 \) and the other in \( V_2 \).

    \item \textbf{Objective}:  
    Maximize the value of the cut, which is the number of edges crossing the partition \( (V_1, V_2) \):
    \[
    \text{Max-Cut} = \max_{V_1 \subseteq V} \left| \text{Cut}(V_1,V_2) \right|
    \]
    where \( V_1 \subseteq V \) is the set of vertices in one partition, and \( V_2 \) is the complement of \( V_1 \) in the partition.

    \item \textbf{Maximization}:  
    The problem is to find the partition \( (V_1^*, V_2^*) \) that maximizes the cut size. Mathematically:
    \[
    (V_1^*, V_2^*) = \arg \max_{(V_1, V_2)} \text{Cut}(V_1, V_2)
    \]
    where \( (V_1^*, V_2^*) \) is the partition that maximizes the number of edges crossing the cut.
\end{enumerate}

\subsection{Reduction from Max-k-SAT to MaxCut}

\begin{theorem}
Any instance of Max-k-SAT can be transformed into an equivalent instance of MaxCut in polynomial time, preserving the optimization structure and solution correspondence.
\end{theorem}

\begin{proof}
\textbf{Construction:}
Given a Max-k-SAT instance with variables $ x_1, \dots, x_n $ and clauses $ C_1, \dots, C_m $, construct a weighted graph $ G(V, E) $ as follows:

\paragraph{Step 1: Variable Consistency Gadgets}
For each variable $ x_i $:
\begin{itemize}
    \item Create two vertices $ v_i $ (representing $ x_i $) and $ \bar{v}_i $ (representing $ \neg x_i $).
    \item Connect $ v_i $ and $ \bar{v}_i $ with an edge of weight $ W = 3m + 1 $. This enforces consistency: in any optimal cut, $ v_i $ and $ \bar{v}_i $ must lie in opposite partitions.
\end{itemize}

\paragraph{Step 2: Clause Gadgets}
For each clause $ C_j = (\ell_1 \lor \ell_2 \lor \dots \lor \ell_k) $:
\begin{itemize}
    \item Introduce two auxiliary vertices $ u_j $ and $ w_j $.
    \item Connect $ u_j $ to each literal vertex in $ C_j $ (i.e., $ v_i $ if $ \ell = x_i $, or $ \bar{v}_i $ if $ \ell = \neg x_i $) with edges of weight 1.
    \item Connect $ w_j $ to each literal vertex in $ C_j $ with edges of weight 1.
    \item Connect $ u_j $ and $ w_j $ with an edge of weight 1.
\end{itemize}

\paragraph{Step 3: Interpretation of Weights}
\begin{itemize}
    \item The large weight $ W = 3m + 1 $ ensures that violating consistency (placing $ v_i $ and $ \bar{v}_i $ in the same partition) incurs a penalty exceeding any possible gain from clauses. Thus, optimal cuts correspond to valid truth assignments.
    \item For clauses, the gadget ensures:
    \begin{itemize}
        \item If at least one literal in $ C_j $ is true (i.e., its vertex is in the opposite partition from $ u_j $), the edges $ (u_j, \text{literal}) $ and $ (w_j, \text{literal}) $ contribute to the cut. This yields a total contribution of $ 3 $ for the gadget (e.g., if $ u_j $ is in $ V_1 $, edges from $ u_j $ to true literals in $ V_2 $ and the edge $ (u_j, w_j) $ contribute 3).
        \item If all literals in $ C_j $ are false, the gadget contributes only $ 2 $ to the cut (edges $ (u_j, w_j) $ and one literal edge).
    \end{itemize}
\end{itemize}

\paragraph{Step 4: Correctness}
Let $ \text{OPT}_{\text{SAT}} $ and $ \text{OPT}_{\text{Cut}} $ denote the optimal values of the Max-k-SAT and MaxCut instances, respectively. By construction:
$$
\text{OPT}_{\text{Cut}} = 3 \cdot \text{OPT}_{\text{SAT}} + nW.
$$
Since $ W = 3m + 1 $, the term $ nW $ dominates, ensuring that any optimal cut must split all $ (v_i, \bar{v}_i) $ pairs. The clause gadget contributions then directly encode satisfied clauses.

\paragraph{Step 5: Complexity}
The graph $ G $ has:
\begin{itemize}
    \item $ 2n + 2m $ vertices (2 per variable, 2 per clause).
    \item $ O(n + km) $ edges (1 edge per variable, $ 2k + 1 $ edges per clause).
\end{itemize}
The construction is polynomial in $ n $ and $ m $.
\end{proof}

\subsection{Reduction from MaxCut to Max-k-SAT}

\begin{theorem}
Any instance of MaxCut can be transformed into an equivalent instance of Max-k-SAT in polynomial time, preserving optimization structure.
\end{theorem}

\begin{proof}
Given a MaxCut instance on a graph $ G(V, E) $, construct a Max-2-SAT instance as follows:

\paragraph{Step 1: Variable Representation}
For each vertex $ v \in V $, create a Boolean variable $ x_v $. The assignment $ x_v = 1 $ (resp. $ x_v = 0 $) encodes placing $ v $ in partition $ V_1 $ (resp. $ V_2 $).

\paragraph{Step 2: Clause Construction}
For each edge $ (u, v) \in E $, add two clauses:
$$
(x_u \lor \neg x_v) \quad \text{and} \quad (\neg x_u \lor x_v).
$$
These clauses are satisfied if and only if $ u $ and $ v $ are in different partitions.

\paragraph{Step 3: Correctness}
The number of satisfied clauses equals twice the size of the cut. Thus, maximizing satisfied clauses corresponds to maximizing the cut value.

\paragraph{Step 4: Complexity}
The reduction generates $ 2|E| $ clauses, which is polynomial in $ |E| $.
\end{proof}

\subsection{Conclusion}
The reductions establish that Max-k-SAT and MaxCut are polynomial-time equivalent. This equivalence provides a theoretical framework to justify extending empirical result, i.e., the logistic saturation conjecture, between the two problems.

\section{Graphs used in the main text} \label{appendix_b}

For completeness, in Fig.~\ref{fig:Fig9} we show all graphs used Sec.~\ref{suppression} and Sec.~\ref{training}. These graphs are randomly generated, subject to the condition that the edges between every pair of vertices are connected with probability $0.5$ with networkx  \cite{hagberg2008exploring}.

\begin{figure}
    \centering
    \includegraphics[width=\linewidth]{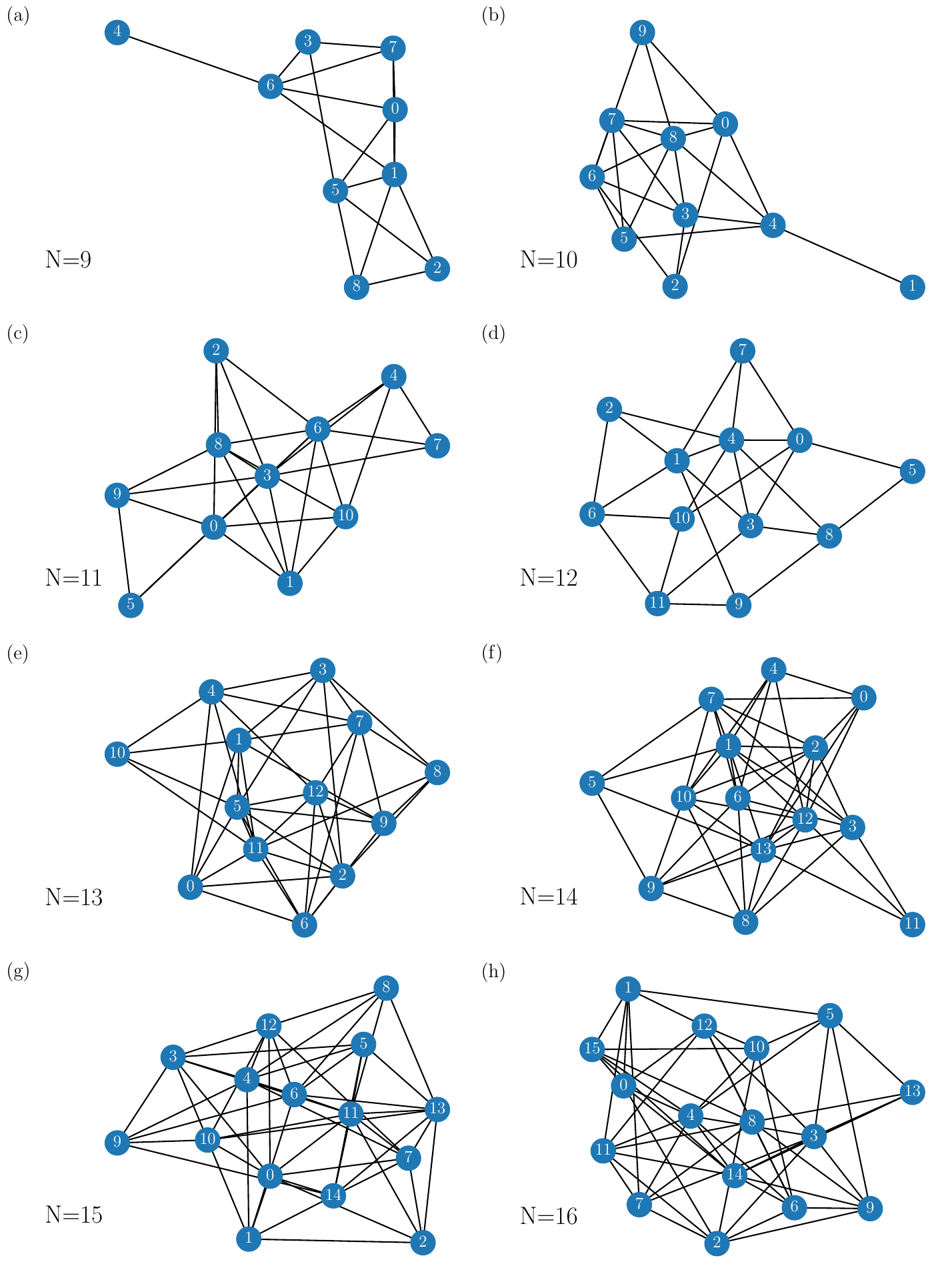}
    \caption{Graphs with vertices ranging from 9 to 16 used in this study}
    \label{fig:Fig9}
\end{figure}

\end{document}